\documentclass[12pt]{article}
\usepackage[latin1]{inputenc}
\usepackage[british]{babel}
\usepackage{cmap}
\usepackage{lmodern}

\usepackage{amssymb, amsmath, amsthm}
\usepackage[a4paper,top=25mm,bottom=25mm,left=25mm,right=25mm]{geometry}
\usepackage{etex}
\usepackage{ragged2e}

\usepackage{authblk} 
\usepackage{pifont}
\usepackage{graphicx}
\usepackage[usenames,dvipsnames,svgnames,table]{xcolor}
\usepackage[figuresright]{rotating}
\usepackage{xtab} 
\usepackage{longtable} 
\usepackage{multirow}
\usepackage{footnote}
\usepackage[stable]{footmisc}
\usepackage{chngpage} 
\usepackage{pdflscape} 
\usepackage{tocbibind} 

\usepackage{pgfplots}
\pgfplotsset{compat=1.14}
\usepackage{setspace}

\makesavenoteenv{tabular}
\usepackage{tabularx}
\usepackage{booktabs}
\usepackage{threeparttable}
\usepackage[referable]{threeparttablex} 
\newcolumntype{R}{>{\raggedleft\arraybackslash}X}
\newcolumntype{L}{>{\raggedright\arraybackslash}X}
\newcolumntype{C}{>{\centering\arraybackslash}X}
\newcolumntype{A}{>{\columncolor{gray!25}}C}
\newcolumntype{a}{>{\columncolor{gray!25}}c}

\newlength{\tablen}

\usepackage{dcolumn} 
\newcolumntype{.}{D{.}{.}{-1}}

\usepackage{tikz}
\usetikzlibrary{arrows, calc, matrix, patterns, positioning, trees}
\usepackage[semicolon]{natbib}
\usepackage[hyphens]{url}
\usepackage{hyperref} 
\hypersetup{
  colorlinks   = true,    
  urlcolor     = blue,    
  linkcolor    = blue,    
  citecolor    = red      
}
\usepackage{microtype}
\usepackage[justification=centering]{caption} 

\usepackage[labelformat=simple]{subcaption}

\DeclareCaptionLabelFormat{parenthesis}{(#2)}
\makeatletter
\renewcommand\p@subfigure{\arabic{figure}.}
\makeatother

\DeclareCaptionLabelFormat{parenthesis}{(#2)}
\makeatletter
\renewcommand\p@subtable{\arabic{table}.}
\makeatother

\usepackage{enumitem}

\setlist[itemize]{leftmargin=2.5\parindent}
\setlist[enumerate]{leftmargin=2.5\parindent}

\theoremstyle{plain}

\newtheorem{corollary}{Corollary}[section]
\newtheorem{lemma}{Lemma}[section]

\newtheorem{proposition}{Proposition}[section]
\newtheorem{theorem}{Theorem}[section]

\theoremstyle{definition}
\newtheorem{axiom}{Axiom}[section]

\newtheorem{definition}{Definition}[section]
\newtheorem{example}{Example}[section]

\theoremstyle{remark}
\newtheorem{notation}{Notation}[section]
\newtheorem{remark}{Remark}[section]

\def\keywords{\vspace{.5em} 
{\noindent \textit{Keywords}:\,}}

\def\JEL{\vspace{.5em} 
{\noindent \textbf{\emph{JEL} classification number}:\,}}

\def\AMS{\vspace{.5em} 
{\noindent \textbf{\emph{MSC} class}:\,}}

\author{\href{https://sites.google.com/site/laszlocsato87}{L\'aszl\'o Csat\'o}\thanks{~e-mail: laszlo.csato@uni-corvinus.hu} }
\affil{Institute for Computer Science and Control, Hungarian Academy of Sciences (MTA SZTAKI) \\
Laboratory on Engineering and Management Intelligence, Research Group of Operations Research and Decision Systems}
\affil{Corvinus University of Budapest (BCE) \\
Department of Operations Research and Actuarial Sciences}
\affil{Budapest, Hungary}
\title{Some impossibilities of ranking in generalized tournaments}
\date{\today}

\begin{document}

\maketitle

\begin{abstract}

In a generalized tournament, players may have an arbitrary number of matches against each other and the outcome of the games is measured on a cardinal scale with a lower and upper bound. An axiomatic approach is applied to the problem of ranking the competitors. Self-consistency requires assigning the same rank for players with equivalent results, while a player showing an obviously better performance than another should be ranked strictly higher. According to order preservation, if two players have the same pairwise ranking in two tournaments where the same players have played the same number of matches, then their pairwise ranking is not allowed to change in the aggregated tournament.
We reveal that these two properties cannot be satisfied simultaneously on this universal domain.

\JEL{C44, D71}

\AMS{91A80, 91B14}

\end{abstract}

\section{Introduction} \label{Sec1}

This paper addresses the problem of tournament ranking when players may have played an arbitrary number of matches against each other, from an axiomatic point of view. For instance, the matches among top tennis players lead to a set of similar data: \emph{Andre Agassi} has played 14 matches with \emph{Boris Becker}, but he has never played against \emph{Bj\"orn Borg} \citep{BozokiCsatoTemesi2016}.
To be more specific, we show the incompatibility of some natural properties.
Impossibility theorems are well-known in the classical theory of social choice \citep{Arrow1950, Gibbard1973, Satterthwaite1975}, but our setting has a crucial difference: the set of agents and the set of alternatives coincide, therefore the transitive effects of 'voting' should be considered \citep{AltmanTennenholtz2008}.
We also allow for cardinal and incomplete preferences as well as ties in the ranking derived.

Several characterizations of ranking methods have been suggested in the literature by providing a set of properties such that they uniquely determine a given method \citep{Rubinstein1980, Bouyssou1992, BouyssouPerny1992, vandenBrinkGilles2003, vandenBrinkGilles2009, SlutzkiVolij2005, SlutzkiVolij2006, Kitti2016}. There are some excellent axiomatic analyses, too \citep{ChebotarevShamis1998a, Gonzalez-DiazHendrickxLohmann2013}.

However, apart from \citet{Csato2018f}, we know only one work discussing impossibility results for ranking the nodes of a directed graph \citep{AltmanTennenholtz2008}, a domain covered by our concept of generalized tournament. We think these theorems are indispensable for a clear understanding of the axiomatic framework. For example, \citet{Gonzalez-DiazHendrickxLohmann2013} have found that most ranking methods violate an axiom called order preservation, but it is not known whether this negative result is caused by a theoretical impossibility or it is only due to some hidden features of the procedures that have been considered.

It is especially a relevant issue because of the increasing popularity of sports rankings \citep{LangvilleMeyer2012}, which is, in a sense, not an entirely new phenomenon, since sports tournaments have motivated some classical works of social choice and voting theory \citep{Landau1895, Zermelo1929, Wei1952}.
For instance, the ranking of tennis players has been addressed from at least three perspectives, with the use of methods from multicriteria decision-making \citep{BozokiCsatoTemesi2016}, network analysis \citep{Radicchi2011}, or statistics \citep{BakerMcHale2014, BakerMcHale2017}. Consequently, the axiomatic approach can be fruitful in the choice of an appropriate sports ranking method. This issue has been discussed in some recent works \citep{Berker2014, Pauly2014, Csato2017d, Csato2018m, Csato2018h, Csato2018j, Csato2018i, Csato2018b, Csato2018l, DagaevSonin2017, Vazirietal2018, Vong2017}, but there is a great scope for future research.



For this purpose, we will place two properties, imported from the social choice literature, in the centre of the discussion.
\emph{Self-consistency} \citep{ChebotarevShamis1997a} requires assigning the same rank for players with equivalent results, furthermore, a player showing an obviously better performance than another should be ranked strictly higher.
\emph{Order preservation}\footnote{~The term order preservation may be a bit misleading since it can suggest that the sequence of matches does not influence the rankings (see \citet[Property~III]{Vazirietal2018}). This requirement obviously holds in our setting.} \citep{Gonzalez-DiazHendrickxLohmann2013} excludes the possibility of rank reversal by demanding the preservation of players' pairwise ranking when two tournaments, where the same players have played the same number of matches, are aggregated. In other words, it is not allowed that player $A$ is judged better both in the first and second half of the season than player $B$, but ranked lower on the basis of the whole season.

Our main result proves the incompatibility of self-consistency and order preservation. This finding gives a theoretical foundation for the observation of \citet{Gonzalez-DiazHendrickxLohmann2013} that most ranking methods do not satisfy order preservation.
Another important message of the paper is that prospective users cannot avoid to take similar impossibilities into account and justify the choice between the properties involved.

The study is structured as follows.
Section~\ref{Sec2} presents the notion of ranking problem and scoring methods.
Section~\ref{Sec3} introduces the property called self-consistency and proves that one type of scoring methods cannot satisfy it. 
Section~\ref{Sec4} defines (strong) order preservation besides some other properties, addresses the compatibility of the axioms and derives a negative result by opposing self-consistency and order preservation.
Section~\ref{Sec5} summarizes our main findings.

\section{The ranking problem and scoring methods} \label{Sec2}

Consider a \emph{set of players} $N = \{ X_1,X_2, \dots, X_n \}$, $n \in \mathbb{N}_+$ and a series of \emph{tournament matrices} $T^{(1)}$, $T^{(2)}$, \dots, $T^{(m)}$ containing information on the paired comparisons of the players.
Their entries are given such that $t_{ij}^{(p)} + t_{ji}^{(p)} = 1$ if players $X_i$ and $X_j$ have played in round $p$ ($1 \leq p \leq m$) and $t_{ij}^{(p)} + t_{ji}^{(p)} = 0$ if they have not played against each other in round $p$.
The simplest definition can be $t_{ij}^{(p)} = 1$ (implying $t_{ji}^{(p)} = 0$) if player $X_i$ has defeated player $X_j$, and $t_{ij}^{(p)} = 0$ (implying $t_{ji}^{(p)} = 1$) if player $X_i$ has lost against player $X_j$ in round $p$. A draw can be represented by $t_{ij}^{(p)} = t_{ji}^{(p)} = 0.5$. The entries may reflect the scores of the players, or other features of the match (e.g. an overtime win has less value than a normal time win), too. 

The tuple $\left( N,T^{(1)}, T^{(2)}, \dots, T^{(m)} \right)$, denoted shortly by $(N,\mathbf{T})$, is called a \emph{general ranking problem}.
The set of general ranking problems with $n$ players ($|N| = n$) is denoted by $\mathcal{T}^n$.

The \emph{aggregated tournament matrix} $A = \sum_{p=1}^m T^{(p)} = \left[ a_{ij} \right] \in \mathbb{R}^{n \times n}$ combines the results of all rounds of the competition.

The pair $(N,A)$ is called a \emph{ranking problem}.
The set of ranking problems with $n$ players ($|N| = n$) is denoted by $\mathcal{R}^n$.
Note that every ranking problem can be associated with several general ranking problems, in this sense, ranking problem is a narrower notion. 

Let $(N,A),(N,A') \in \mathcal{R}^n$ be two ranking problems with the same player set $N$. The \emph{sum} of these ranking problems is $(N,A+A') \in \mathcal{R}^n$.
For example, the ranking problems can contain the results of matches in the first and second half of the season, respectively.

Any ranking problem $(N,A)$ has a skew-symmetric \emph{results matrix} $R = A - A^\top = \left[ r_{ij} \right] \in \mathbb{R}^{n \times n}$ and a symmetric \emph{matches matrix} $M = A + A^\top = \left[ m_{ij} \right] \in \mathbb{N}^{n \times n}$. $m_{ij}$ is the number of matches between players $X_i$ and $X_j$, whose outcome is given by $r_{ij}$. Matrices $R$ and $M$ also determine the aggregated tournament matrix through $A = (R + M)/2$, so any ranking problem $(N,A) \in \mathcal{R}^n$ can be denoted analogously by $(N,R,M)$ with the restriction $|r_{ij}| \leq m_{ij}$ for all $X_i,X_j \in N$.
Despite description with results and matches matrices is not parsimonious, this notation will turn out to be useful.

A \emph{general scoring method} is a function $g:\mathcal{T}^n \to \mathbb{R}^n$. Several procedures have been suggested in the literature, see \citet{ChebotarevShamis1998a} for an overview of them. A special type of general scoring methods is the following.

\begin{definition} \label{Def21}
\emph{Individual scoring method} \citep{ChebotarevShamis1999}: A general scoring method $g:\mathcal{T}^n \to \mathbb{R}^n$ is called \emph{individual scoring method} if it is based on individual scores, that is, there exist functions $\phi$ and $\delta$ such that for any general ranking problem $(N,\mathbf{T}) \in \mathcal{T}^n$, the corresponding score vector $\mathbf{s} = g(N,\mathbf{T})$ can be expressed as $\mathbf{s} = \delta(\mathbf{s}^{(1)},\mathbf{s}^{(2)}, \dots, \mathbf{s}^{(m)})$, where the partial score vectors $\mathbf{s}^{(p)} = \phi(N,T^{(p)})$ depend solely on the tournament matrix $T^{(p)}$ of round $p$ for all $p = 1,2, \dots, m$.
\end{definition}

A \emph{scoring method} is a function $f:\mathcal{R}^n \to \mathbb{R}^n$. Any scoring method can also be regarded as a general scoring method -- by using the aggregated tournament matrix instead of the whole series of tournament matrices --, therefore some articles only consider scoring methods \citep{Kitti2016, SlutzkiVolij2005}. \citet{Gonzalez-DiazHendrickxLohmann2013} give a thorough axiomatic analysis of certain scoring methods.

In other words, scoring methods initially aggregate the tournament matrices and then rank the players by their scores, while individual scoring methods first give scores to the players in each round and then aggregate them.

\section{An argument against the use of individual scoring methods} \label{Sec3}

In this section, some properties of general scoring methods are presented, which will highlight an important failure of individual scoring methods.

\subsection{Universal invariance axioms} \label{Sec31}

\begin{axiom} \label{Axiom31}
\emph{Anonymity} ($ANO$): 
Let $(N,\mathbf{T}) \in \mathcal{T}^n$ be a general ranking problem, $\sigma: \{ 1,2, \dots, m \} \rightarrow \{ 1,2, \dots, m \}$ be a permutation on the set of rounds, and $\sigma(N,\mathbf{T}) \in \mathcal{T}^n$ be the ranking problem obtained from $(N,\mathbf{T})$ by permutation $\sigma$.
General scoring method $g: \mathcal{T}^n \to \mathbb{R}^n$ is \emph{anonymous} if $g_i (N,\mathbf{T}) = g_i \left( \sigma(N,\mathbf{T}) \right)$ for all $X_i \in N$.
\end{axiom}

Anonymity implies that any reindexing of the rounds (tournament matrices) preserves the scores of the players.

\begin{axiom} \label{Axiom32}
\emph{Neutrality} ($NEU$): 
Let $(N,\mathbf{T}) \in \mathcal{T}^n$ be a general ranking problem, $\sigma: N \rightarrow N$ be a permutation on the set of players, and $(\sigma(N),\mathbf{T}) \in \mathcal{T}^n$ be the ranking problem obtained from $(N,\mathbf{T})$ by permutation $\sigma$.
General scoring method $g: \mathcal{T}^n \to \mathbb{R}^n$ is \emph{neutral} if $g_i(N,\mathbf{T}) = g_{\sigma(i)} (\sigma(N),\mathbf{T})$ for all $X_i \in N$.
\end{axiom}

Neutrality means that the scores are independent of the labelling of the players.

\subsection{Self-consistency} \label{Sec32}

Now we want to formulate a further requirement on the ranking of the players by answering the following question: \emph{When is player $X_i$ undeniably better than player $X_j$?}
There are two such plausible cases: (1) if player $X_i$ has achieved better results against the same opponents; (2) if player $X_i$ has achieved the same results against stronger opponents. Consequently, player $X_i$ should also be judged better if he/she has achieved better results against stronger opponents than player $X_j$. Furthermore, since (general) scoring methods allow for ties in the ranking, player $X_i$ should have the same rank as player $X_j$ if he/she has achieved the same results against opponents with the same strength.

In order to apply these principles, both the results and strengths of the players should be measured. Results can be extracted from the tournament matrices $T^{(p)}$. Strengths of the players can be obtained from their scores according to the (general) scoring method used, hence the name of the implied axiom is \emph{self-consistency}.
It has been introduced in \citet{ChebotarevShamis1997a}, and extensively discussed by \citet{Csato2018f}.

\begin{definition} \label{Def31}
\emph{Opponent multiset}:
Let $(N,\mathbf{T}) \in \mathcal{T}^n$ be a general ranking problem. The \emph{opponent multiset}\footnote{~\emph{Multiset} is a generalization of the concept of set allowing for multiple instances of its elements.}
of player $X_i$ is $O_i$, which contains $m_{ij}$ instances of $X_j$.
\end{definition}

Players of the opponent multiset $O_i$ are called the \emph{opponents} of player $X_i$.

\begin{notation} \label{Not31}
Consider the ranking problem $(N,T^{(p)}) \in \mathcal{T}^n$ given by restricting a general ranking problem to its $p$th round. Let $X_i, X_j \in N$ be two different players and $h^{(p)}: O_i^{(p)} \leftrightarrow O_j^{(p)}$ be a one-to-one correspondence between the opponents of $X_i$ and $X_j$ in round $p$, consequently, $|O_i^{(p)}| =  |O_j^{(p)}|$.
Then $\mathfrak{h}^{(p)}: \{k: X_k \in O_i^{(p)} \} \leftrightarrow \{\ell: X_\ell \in O_j^{(p)} \}$ is given by $X_{\mathfrak{h}^{(p)}(k)} = h^{(p)}(X_k)$.
\end{notation}

\begin{axiom} \label{Axiom33}
\emph{Self-consistency} ($SC$) \citep{ChebotarevShamis1997a}:
A general scoring method $g: \mathcal{T}^n \to \mathbb{R}^n$ is called \emph{self-consistent} if the following implication holds for any general ranking problem $(N,\mathbf{T}) \in \mathcal{T}^n$ and for any players $X_i,X_j \in N$:
if there exists a one-to-one mapping $h^{(p)}$ from $O^{(p)}_i$ onto $O^{(p)}_j$ such that $t_{ik}^{(p)} \geq t_{j \mathfrak{h}^{(p)}(k)}^{(p)}$ and $g_k(N,\mathbf{T}) \geq g_{\mathfrak{h}^{(p)}(k)}(N,\mathbf{T})$ for all $p = 1,2, \dots ,m$ and $X_k \in O_i^{(p)}$, then
$f_i(N,R,M) \geq f_{j}(N,R,M)$, furthermore, $f_i(N,R,M) > f_{j}(N,R,M)$ if $t_{ik}^{(p)} > t_{j \mathfrak{h}^{(p)}(k)}^{(p)}$ or $g_k(N,\mathbf{T}) > g_{\mathfrak{h}^{(p)}(k)}(N,\mathbf{T})$ for at least one $1 \leq p \leq m$ and $X_k \in O_i^{(p)}$.
\end{axiom}


\subsection{Individual scoring methods and self-consistency} \label{Sec33}

In this part, it will be proved that an anonymous and neutral individual scoring method cannot satisfy self-consistency, which is a natural fairness requirement, thus it is enough to focus on ranking problems and scoring methods.
For this purpose, the example below will be used.

\begin{figure}[htbp]
\centering
\caption{The general ranking problem of Example~\ref{Examp31}}
\label{Fig31}
  
\begin{subfigure}{.33\textwidth}
  \centering
  \subcaption{$(N,T^{(1)})$}
  \label{Fig31a}
\begin{tikzpicture}[scale=1, auto=center, transform shape, >=triangle 45]
\tikzstyle{every node}=[draw,shape=rectangle]; 
  \node (n1) at (135:2) {$X_1$};
  \node (n2) at (45:2)  {$X_2$};
  \node (n3) at (315:2) {$X_3$};
  \node (n4) at (225:2) {$X_4$};

  \draw [->] (n1) -- (n4);
\end{tikzpicture}
\end{subfigure}
\begin{subfigure}{.33\textwidth}
  \centering
  \subcaption{$(N,T^{(2)})$}
  \label{Fig31b}
\begin{tikzpicture}[scale=1, auto=center, transform shape, >=triangle 45]
\tikzstyle{every node}=[draw,shape=rectangle];
  \node (n1) at (135:2) {$X_1$};
  \node (n2) at (45:2)  {$X_2$};
  \node (n3) at (315:2) {$X_3$};
  \node (n4) at (225:2) {$X_4$};

  \draw (n1) -- (n2);
  \draw (n4) -- (n3);
\end{tikzpicture}
\end{subfigure}
\begin{subfigure}{.33\textwidth}
  \centering
  \subcaption{$(N,\mathbf{T})$}
  \label{Fig31c}
\begin{tikzpicture}[scale=1, auto=center, transform shape, >=triangle 45]
\tikzstyle{every node}=[draw,shape=rectangle];
  \node (n1) at (135:2) {$X_1$};
  \node (n2) at (45:2)  {$X_2$};
  \node (n3) at (315:2) {$X_3$};
  \node (n4) at (225:2) {$X_4$};

  \foreach \from/\to in {n1/n2,n3/n4}
    \draw (\from) -- (\to);
  \draw [->] (n1) -- (n4);
\end{tikzpicture}
\end{subfigure}
\end{figure}
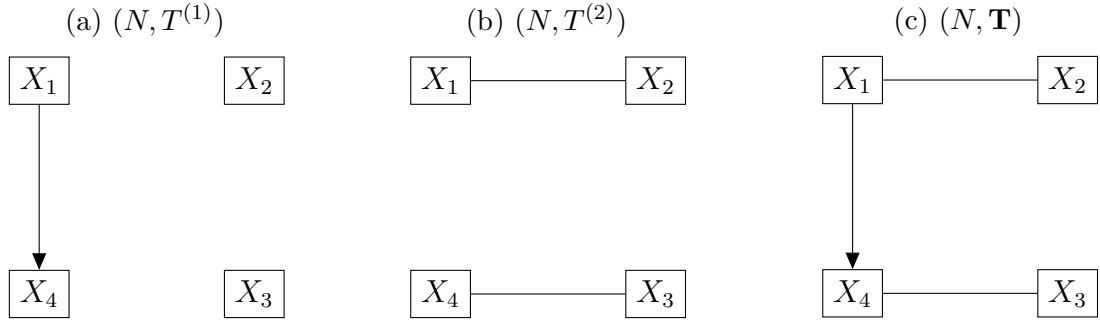

\begin{example} \label{Examp31}
Let $\left( N,T^{(1)},T^{(2)} \right) \in \mathcal{T}^4$ be a general ranking problem describing a tournament with two rounds. 

It is shown in Figure~\ref{Fig31}: a directed edge from node $X_i$ to $X_j$ indicates a win of player $X_i$ over $X_j$ (and a loss of $X_j$ against $X_i$), while an undirected edge from node $X_i$ to $X_j$ represents a drawn match between the two players. This representation will be used in further examples, too.

So, player $X_1$ has defeated $X_4$ in the first round (Figure~\ref{Fig31a}), while players $X_2$ and $X_3$ have not played. In the second round, players $X_1$ and $X_2$, as well as players $X_3$ and $X_4$ have drawn (Figure~\ref{Fig31b}). The whole tournament is shown in Figure~\ref{Fig31c}.
\end{example}

According to the following result, at least one property from the set of $ANO$, $NEU$ and $SC$ will be violated by any individual scoring method. 

\begin{proposition} \label{Prop31}
There exists no anonymous and neutral individual scoring method satisfying self-consistency.
\end{proposition}

\begin{proof}
Let $g: \mathcal{T}^n \to \mathbb{R}^n$ be an anonymous and neutral individual scoring method.
Consider Example~\ref{Examp31}.
$ANO$ and $NEU$ imply that $g_2(N,T^{(1)}) = g_3(N,T^{(1)})$ and $g_2(N,T^{(2)}) = g_3(N,T^{(2)})$, therefore
\begin{equation} \label{eq1}
g_2(N,\mathbf{T}) = \delta \left( g_2(N,T^{(1)}), g_2(N,T^{(2)}) \right) = \delta \left( g_3(N,T^{(1)}), g_3(N,T^{(2)}) \right) = g_3(N,\mathbf{T}).
\end{equation}

Note that $O_1^{(1)} = \{ X_4 \}$, $O_1^{(1)} = \{ X_2 \}$ and $O_4^{(1)} = \{ X_1 \}$, $O_4^{(2)} = \{ X_3 \}$.
Take the one-to-one correspondences $h_{14}^{(1)}: O_1^{(1)} \leftrightarrow O_4^{(1)}$ such that $h_{14}^{(1)}(X_4)=X_1$ and $h_{14}^{(2)}: O_1^{(2)} \leftrightarrow O_4^{(2)}$ such that $h_{14}^{(2)}(X_2)=X_3$. Now $t_{12}^{(2)} = t_{43}^{(2)}$ since the corresponding matches resulted in draws. Furthermore, $t_{14}^{(1)} \neq t_{41}^{(1)}$ since the value of a win and a loss should be different.
It can be assumed without loss of generality that $t_{14}^{(1)} > t_{41}^{(1)}$.
Suppose that $g_1(N,\mathbf{T}) \leq g_4(N,\mathbf{T})$. Then players $X_1$ and $X_4$ have a draw against a player with the same strength ($X_2$ and $X_3$, respectively), but $X_1$ has defeated $X_4$, so it has a better result against a not weaker opponent. Therefore, self-consistency (Axiom~\ref{Axiom33}) implies $g_1(N,\mathbf{T}) > g_4(N,\mathbf{T})$, which is a contradiction, thus $g_1(N,\mathbf{T}) > g_4(N,\mathbf{T})$ holds.
 
However, $O_2^{(1)} = \emptyset$, $O_2^{(2)} = \{ X_1 \}$ and $O_3^{(1)} = \emptyset$, $O_3^{(2)} = \{ X_4 \}$. Consider the unique one-to-one correspondence $h_{14}^{(2)}: O_2^{(2)} \leftrightarrow O_3^{(2)}$, which -- together with $t_{21}^{(2)} = t_{34}^{(2)}$ (the two draws should be represented by the same number) and $g_1(N,\mathbf{T}) > g_4(N,\mathbf{T})$ -- leads to $g_2(N,\mathbf{T}) > g_3(N,\mathbf{T})$ because player $X_2$ has achieved the same result against a stronger opponent than player $X_3$. In other words, $SC$ requires the draw of $X_2$ to be more valuable than the draw of $X_3$, but it cannot be reflected by any individual scoring method $g$ according to \eqref{eq1}.
\end{proof}

\section{The case of ranking problems and scoring methods} \label{Sec4}

According to Proposition~\ref{Prop31}, only the procedure underlying scoring methods can be compatible with self-consistency.
Therefore, this section will focus on scoring methods.

\subsection{Axioms of invariance with respect to the results matrix} \label{Sec41}

Let $O \in \mathbb{R}^{n \times n}$ be the matrix with all of its entries being zero.

\begin{axiom} \label{Axiom41}
\emph{Symmetry} ($SYM$) \citep{Gonzalez-DiazHendrickxLohmann2013}:
Let $(N,R,M) \in \mathcal{R}^n$ be a ranking problem such that $R=O$.
Scoring method $f: \mathcal{R}^n \to \mathbb{R}^n$ is \emph{symmetric} if $f_i(N,R,M) = f_j(N,R,M)$ for all $X_i, X_j \in N$.
\end{axiom}

According to symmetry, if all paired comparisons (but not necessarily all matches in each round) between the players result in a draw, then all players will have the same score.

\begin{axiom} \label{Axiom42}
\emph{Inversion} ($INV$) \citep{ChebotarevShamis1998a}:
Let $(N,R,M) \in \mathcal{R}^n$ be a ranking problem.
Scoring method $f: \mathcal{R}^n \to \mathbb{R}^n$ is \emph{invertible} if $f_i(N,R,M) \geq f_j(N,R,M) \iff f_i(N,-R,M) \leq f_j(N,-R,M)$ for all $X_i, X_j \in N$.
\end{axiom}

Inversion means that taking the opposite of all results changes the ranking accordingly. It establishes a uniform treatment of victories and losses.

\begin{corollary} \label{Col41}
Let $f: \mathcal{R}^n \to \mathbb{R}^n$ be a scoring method satisfying $INV$. Then for all $X_i, X_j \in N$: $f_i(N,R,M) > f_j(N,R,M) \iff f_i(N,-R,M) < f_j(N,-R,M)$.
\end{corollary}

The following result has been already mentioned by \citet[p.~150]{Gonzalez-DiazHendrickxLohmann2013}.

\begin{corollary} \label{Col42}
$INV$ implies $SYM$.
\end{corollary}

It seems to be difficult to argue against symmetry. However, scoring methods based on right eigenvectors \citep{Wei1952, SlutzkiVolij2005, SlutzkiVolij2006, Kitti2016} violate inversion. 

\subsection{Properties of independence} \label{Sec42}

The next axiom deals with the effects of certain changes in the aggregated tournament matrix $A$.

\begin{axiom} \label{Axiom43}
\emph{Independence of irrelevant matches} ($IIM$) \citep{Gonzalez-DiazHendrickxLohmann2013}:
Let $(N,A),(N,A') \in \mathcal{R}^n$ be two ranking problems and $X_i,X_j,X_k, X_\ell \in N$ be four different players such that $(N,A)$ and $(N,A')$ are identical but $a_{k \ell} \neq a'_{k \ell}$.
Scoring method $f: \mathcal{R}^n \to \mathbb{R}^n$ is called \emph{independent of irrelevant matches} if $f_i(N,A) \geq f_j(N,A) \Rightarrow f_i(N,A') \geq f_j(N,A')$.
\end{axiom}

$IIM$ means that 'remote' matches -- not involving players $X_i$ and $X_j$ -- do not affect the pairwise ranking of players $X_i$ and $X_j$.



Independence of irrelevant matches seems to be a powerful property. \citet{Gonzalez-DiazHendrickxLohmann2013} state that '\emph{when players have different opponents (or face opponents with different intensities), $IIM$ is a property one would rather not have}'. \citet{Csato2018f} argues on an axiomatic basis against $IIM$.

The rounds of a given tournament can be grouped arbitrarily. Therefore, the following property makes much sense. 

\begin{axiom} \label{Axiom44}
\emph{Order preservation} ($OP$) \citep{Gonzalez-DiazHendrickxLohmann2013}:
Let $(N,A),(N,A') \in \mathcal{R}^n$ be two ranking problems where all players have played $m$ matches and $X_i, X_j \in N$ be two different players.
Let $f: \mathcal{R}^n \to \mathbb{R}^n$ be a scoring method such that $f_i(N,A) \geq f_j(N,A)$ and $f_i(N,A') \geq f_j(N,A')$.\footnote{~\citet{Gonzalez-DiazHendrickxLohmann2013} formally introduce a stronger version of this axiom since only $X_i$ and $X_j$ should have the same number of matches in the two ranking problems. However, in the counterexample of \citet{Gonzalez-DiazHendrickxLohmann2013}, which shows the violation of $OP$ by several ranking methods, all players have played the same number of matches.}
$f$ satisfies \emph{order preservation} if $f_i(N,A+A') \geq f_j(N,A+A')$, furthermore, $f_i(N,A+A') > f_j(N,A+A')$ if $f_i(N,A) > f_j(N,A)$ or $f_i(N,A') > f_j(N,A')$.
\end{axiom}

$OP$ is a relatively restricted version of combining ranking problems, which implies that if player $X_i$ is not worse than player $X_j$ on the basis of some rounds as well as on the basis of another set of rounds such that all players have played in each round (so they have played the same number of matches altogether), then this pairwise ranking should hold after the two distinct set of rounds are considered jointly.

One can consider a stronger version of order preservation, too.

\begin{axiom} \label{Axiom45}
\emph{Strong order preservation} ($SOP$) \citep{vandenBrinkGilles2009}:
Let $(N,A),(N,A') \in \mathcal{R}^n$ be two ranking problems and $X_i, X_j \in N$ be two players.
Let $f: \mathcal{R}^n \to \mathbb{R}^n$ be a scoring method such that $f_i(N,A) \geq f_j(N,A)$ and $f_i(N,A') \geq f_j(N,A')$.
$f$ satisfies \emph{strong order preservation} if $f_i(N,A+A') \geq f_j(N,A+A')$, furthermore, $f_i(N,A+A') > f_j(N,A+A')$ if $f_i(N,A) > f_j(N,A)$ or $f_i(N,A') > f_j(N,A')$.
\end{axiom}

In contrast to order preservation, $SOP$ does not contain any restriction on the number of matches of the players in the ranking problems to be aggregated.

\begin{corollary} \label{Col43}
$SOP$ implies $OP$.
\end{corollary}

It will turn out that the weaker property, order preservation has still unfavourable implications.

\subsection{Relations among the axioms} \label{Sec44}

In this part, some links among symmetry, inversion, independence of irrelevant matches, and (strong) order preservation will be revealed.


\begin{remark} \label{Rem41}
$SYM$ and $OP$ ($SOP$) imply $INV$.
\end{remark}

\begin{proof}
Consider a ranking problem $(N,R,M) \in \mathcal{R}^n$ where $f_i(N,R,M) \geq f_j(N,R,M)$ for players $X_i,X_j \in N$. 
If $f_i(N,-R,M) > f_j(N,-R,M)$, then $f_i(N,O,2M) > f_j(N,O,2M)$ due to $OP$, which contradicts to $SYM$. So $f_i(N,-R,M) \leq f_j(N,-R,M)$ holds.
\end{proof}

It turns out that $IIM$ is also closely connected to $SOP$.

\begin{proposition} \label{Prop41}
A scoring method satisfying $NEU$, $SYM$ and $SOP$ meets $IIM$.
\end{proposition}

\begin{proof}
Assume to the contrary, and let $(N,R,M) \in \mathcal{R}^n$ be a ranking problem, $f: \mathcal{R}^n \to \mathbb{R}^n$ be a scoring method satisfying $NEU$, $SYM$, and $SOP$, and $X_i, X_j, X_k, X_\ell \in N$ be four different players such that $f_i(N,R,M) \geq f_j(N,R,M)$, and $(N,R',M') \in \mathcal{R}^n$ is identical to $(N,R,M)$ except for the result $r'_{k \ell}$ and number of matches $m'_{k \ell}$ between players $X_k$ and $X_\ell$, where $f_i(N,R',M') < f_j(N,R',M')$.

According to Remark~\ref{Rem41}, $f$ satisfies $INV$, hence $f_i(N,-R,M) \leq f_j(N,-R,M)$.
Denote by $\sigma: N \rightarrow N$ the permutation $\sigma(X_i) = X_j$, $\sigma(X_j) = X_i$, and $\sigma(X_k) = X_k$ for all $X_k \in N \setminus \{ X_i,X_j \}$. Neutrality leads to in $f_i \left[ \sigma(N,R,M) \right] \leq f_j \left[ \sigma(N,R,M) \right]$, and $f_i \left[ \sigma(N,-R',M') \right] < f_j \left[ \sigma(N,-R',M') \right]$ due to inversion and Corollary~\ref{Col41}. With the notations $R'' = \sigma(R) - \sigma(R') - R + R' = O$ and $M'' = \sigma(M) + \sigma(M') + M + M'$, we get
\[
(N,R'',M'') = \sigma(N,R,M) + \sigma(N,-R',M') + (N,-R,M) + (N,R',M').
\]
Symmetry implies $f_i(N,R'',M'') = f_j(N,R'',M'')$ since $R'' = O$, but $f_i(N,R'',M'') < f_j(N,R'',M'')$ from strong order preservation, which is a contradiction. 
\end{proof}


It remains to be seen whether $NEU$, $SYM$, and $SOP$ are all necessary for Proposition~\ref{Prop41}.

\begin{lemma} \label{Lemma41}
$NEU$, $SYM$, and $SOP$ are logically independent axioms with respect to the implication of $IIM$.
\end{lemma}

\begin{proof}
It is shown that there exist scoring methods, which satisfy exactly two properties from the set $NEU$, $SYM$, and $SOP$, but violate the third and does not meet $IIM$, too:
\begin{enumerate}[label=\fbox{\arabic*}]
\item
$SYM$ and $SOP$: the sum of the results of the 'previous' player, $f_i(N,R,M) = \sum_{j=1}^n r_{i-1,j}$ for all $X_i \in N \setminus \{ X_1 \}$ and $f_1(N,R,M) = \sum_{j=1}^n r_{n,j}$;
\item
$NEU$ and $SOP$: maximal number of matches of other players, $f_i(N,R,M) = \max \{ \sum_{k=1}^n m_{jk}: X_j \neq X_i \}$;\footnote{~It is worth to note that the maximal number of own matches satisfies $NEU$, $SOP$, and $IIM$.}
\item
$NEU$ and $SYM$: aggregated sum of the results of opponents, $f_i(N,R,M) = \sum_{X_j \in O_i} \sum_{k=1}^n r_{jk}$.
\end{enumerate}
\end{proof}

Proposition~\ref{Prop41} helps in deriving another impossibility statement.

\begin{proposition} \label{Prop42}
There exists no scoring method that satisfies neutrality, symmetry, strong order preservation and self-consistency.
\end{proposition}

\begin{proof}
According to Proposition~\ref{Prop41}, $NEU$, $SYM$ and $SOP$ imply $IIM$. 
\citet[Theorem~3.1]{Csato2018f} has shown that $IIM$ and $SC$ cannot be met at the same time.
\end{proof}

\subsection{A basic impossibility result} \label{Sec45}

The four axioms of Proposition~\ref{Prop42} are not independent despite Lemma~\ref{Lemma41}. 
However, a much stronger statement can be obtained by eliminating neutrality and symmetry, which also allows for a weakening of strong order preservation by using order preservation. Note that substituting an axiom with a weaker one in an impossibility statement leads to a stronger result.

We will use a generalized tournament with four players for this purpose.

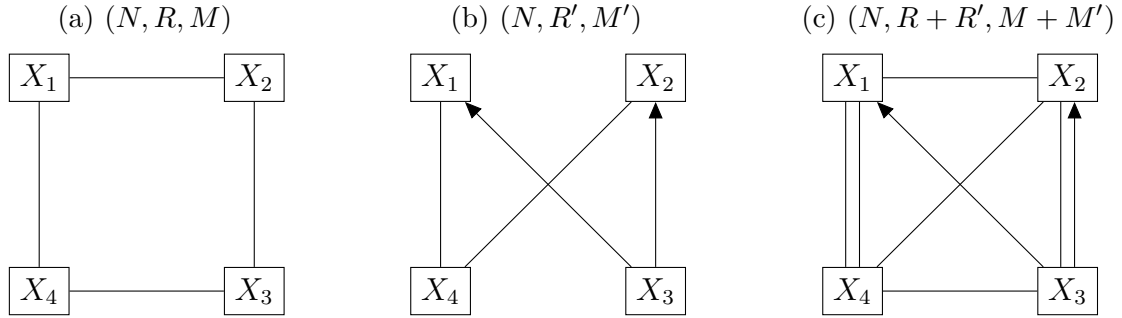
\begin{figure}[htbp]
\centering
\caption{The ranking problems of Example~\ref{Examp41}}
\label{Fig41}
  
\begin{subfigure}{.33\textwidth}
  \centering
  \subcaption{$(N,R,M)$}
  \label{Fig41a}
\begin{tikzpicture}[scale=1, auto=center, transform shape, >=triangle 45]
\tikzstyle{every node}=[draw,shape=rectangle]; 
  \node (n1) at (135:2) {$X_1$};
  \node (n2) at (45:2)  {$X_2$};
  \node (n3) at (315:2) {$X_3$};
  \node (n4) at (225:2) {$X_4$};

  \foreach \from/\to in {n1/n2,n1/n4,n2/n3,n3/n4}
    \draw (\from) -- (\to);
\end{tikzpicture}
\end{subfigure}
\begin{subfigure}{.33\textwidth}
  \centering
  \subcaption{$(N,R',M')$}
  \label{Fig41b}
\begin{tikzpicture}[scale=1, auto=center, transform shape, >=triangle 45]
\tikzstyle{every node}=[draw,shape=rectangle];
  \node (n1) at (135:2) {$X_1$};
  \node (n2) at (45:2)  {$X_2$};
  \node (n3) at (315:2) {$X_3$};
  \node (n4) at (225:2) {$X_4$};

  \foreach \from/\to in {n1/n4,n2/n4}
    \draw (\from) -- (\to);
  \draw [->] (n3) -- (n1);
  \draw [->] (n3) -- (n2);
\end{tikzpicture}
\end{subfigure}
\begin{subfigure}{.33\textwidth}
	\centering
	\caption{$(N,R+R',M+M')$}
	\label{Fig41c}
\begin{tikzpicture}[scale=1, auto=center, transform shape, >=triangle 45]
\tikzstyle{every node}=[draw,shape=rectangle];
  \node (n1) at (135:2) {$X_1$};
  \node (n2) at (45:2)  {$X_2$};
  \node (n3) at (315:2) {$X_3$};
  \node (n4) at (225:2) {$X_4$};

  \foreach \from/\to in {n1/n2,n2/n4,n3/n4}
    \draw (\from) -- (\to);
  \draw [->] (n3) -- (n1);
\draw[transform canvas={xshift=-0.5ex}](n2) -- (n3);
\draw[->,transform canvas={xshift=0.5ex}](n3) -- (n2);
\draw[transform canvas={xshift=-0.5ex}](n1) -- (n4);
\draw[transform canvas={xshift=0.5ex}](n1) -- (n4);
\end{tikzpicture}
\end{subfigure}
\end{figure}

\begin{example} \label{Examp41}
Let $(N,R,M), (N,R',M') \in \mathcal{R}^4$ be two ranking problems.
They are shown in Figure~\ref{Fig41}: in the first tournament described by $(N,R,M)$, matches between players $X_1$ and $X_2$, $X_1$ and $X_4$, $X_2$ and $X_3$, $X_3$ and $X_4$ all resulted in draws (see Figure~\ref{Fig41a}). On the other side, in the second tournament, described by $(N,R',M')$, players $X_1$ and $X_2$ have lost against $X_3$ and drawn against $X_4$ (see Figure~\ref{Fig41b}).
The two ranking problems can be summed in $(N,R'',M'') \in \mathcal{R}^4$ such that $R'' = R + R'$ and $M'' = M + M'$ (see Figure~\ref{Fig41c}).
\end{example}

\begin{theorem} \label{Theo41}
There exists no scoring method that satisfies order preservation and self-consistency.
\end{theorem}

\begin{proof}
Assume to the contrary that there exists a self-consistent scoring method $f: \mathcal{R}^n \to \mathbb{R}^n$ satisfying order preservation.
Consider Example~\ref{Examp41}.

\begin{enumerate}[label=\emph{\Roman*}.]
\item
Take the ranking problem $(N,R,M)$. Note that $O_1 = O_3 = \{ X_2, X_4 \}$ and $O_2 = O_4 = \{ X_1, X_3 \}$.
\begin{enumerate}[label=\emph{\alph*})]
\item
Consider the identity one-to-one correspondences $h_{13}: O_1 \leftrightarrow O_3$ and $h_{31}: O_3 \leftrightarrow O_1$ such that $h_{13}(X_2) = h_{31}(X_2) = X_2$ and $h_{13}(X_4) = h_{31}(X_4) = X_4$. Since $r_{12} = r_{32} = 0$ and $r_{14} = r_{34} = 0$, players $X_1$ and $X_3$ have the same results against the same opponents, hence $f_1(N,R,M) = f_3(N,R,M)$ from $SC$.

\item
Consider the identity one-to-one correspondences $h_{24}: O_2 \leftrightarrow O_4$ and $h_{42}: O_4 \leftrightarrow O_2$. Since $r_{21} = r_{41} = 0$ and $r_{23} = r_{43} = 0$, players $X_2$ and $X_4$ have the same results against the same opponents, hence $f_2(N,R,M) = f_4(N,R,M)$ from $SC$.

\item
Suppose that $f_2(N,R,M) > f_1(N,R,M)$, which implies $f_4(N,R,M) > f_3(N,R,M)$.
Consider the one-to-one mapping $h_{12}: O_1 \leftrightarrow O_2$, where $h_{12}(X_2) = X_1$ and $h_{12}(X_4) = X_3$. Since $r_{12} = r_{21} = 0$ and $r_{14} = r_{23} = 0$, player $X_1$ has the same results against stronger opponents compared to $X_2$, hence $f_1(N,R,M) > f_2(N,R,M)$ from $SC$, which is a contradiction.

\item
An analogous argument shows that $f_1(N,R,M) > f_2(N,R,M)$ cannot hold.
\end{enumerate}
Therefore, self-consistency leads to $f_1(N,R,M) = f_2(N,R,M) = f_3(N,R,M) = f_4(N,R,M)$ in the first ranking problem.

\item
Take the ranking problem $(N,R',M')$. Note that $O_1' = O_2' = \{ X_3, X_4 \}$ and $O_3' = O_4' = \{ X_1, X_2 \}$.
\begin{enumerate}[label=\emph{\alph*})]
\item
Consider the identity one-to-one correspondences $h_{12}': O_1' \leftrightarrow O_2'$ and $h_{21}': O_2' \leftrightarrow O_1'$. Since $r_{13}' = r_{23}' = -1$ and $r_{14}' = r_{24}' = 0$, players $X_1$ and $X_2$ have the same results against the same opponents, hence $f_1(N,R',M') = f_2(N,R',M')$ from $SC$.

\item
Consider the identity one-to-one correspondence $h_{34}': O_3' \leftrightarrow O_4'$. Since $1 = r_{31}' > r_{41}' = 0$ and $1 = r_{32}' > r_{42}' = 0$, player $X_3$ has better results against the same opponents compared to $X_4$, hence $f_3(N,R',M) > f_4(N,R',M)$ from $SC$.
\end{enumerate}
Thus self-consistency leads to $f_1(N,R',M') = f_2(N,R',M')$ and $f_3(N,R',M') > f_4(N,R',M')$ in the second ranking problem.

\item
Take the sum of these two ranking problems, the ranking problem $(N,R'',M'')$.

Suppose that $f_1(N,R'',M'') \geq f_2(N,R'',M'')$. Consider the one-to-one mappings $g_{21}: O_2 \leftrightarrow O_1$ and $g_{21}': O_2' \leftrightarrow O_1'$ such that $g_{21}(X_1) = X_2$, $g_{21}(X_3) = X_4$ and $g_{21}'(X_3) = X_3$, $g_{21}'(X_4) = X_4$.
Since $r_{21} = r_{12} = 0$, $r_{23} = r_{14} = 0$ and $r_{23}' = r_{13}' = -1$, $r_{24}' = r_{14}' = 0$, player $X_2$ has the same results against stronger opponents compared to $X_1$, hence $f_2(N,R'',M'') > f_1(N,R'',M'')$ from $SC$, which leads to a contradiction.

To summarize, self-consistency results in $f_1(N,R'',M') < f_2(N,R'',M'')$, however, order preservation implies $f_1(N,R'',M'') = f_2(N,R'',M'')$ as all players have played two matches in $(N,R',M')$ and $(N,R',M')$, respectively, which is impossible.
\end{enumerate}

Therefore, it has been derived that no scoring method can meet $OP$ and $SC$ simultaneously on the universal domain of $\mathcal{R}^n$.
\end{proof}

Theorem~\ref{Theo41} is a serious negative result: by accepting self-consistency, the ranking method cannot be required to preserve two players' pairwise ranking when some ranking problems, where all players have played the same number of matches, are aggregated.

\begin{figure}[htbp]
\centering
\caption{The ranking problem of Example~\ref{Examp42}}
\label{Fig42}
\begin{tikzpicture}[scale=1, auto=center, transform shape, >=triangle 45]
\tikzstyle{every node}=[draw,shape=rectangle];
  \node (n1) at (135:2) {$X_1$};
  \node (n2) at (45:2)  {$X_2$};
  \node (n3) at (315:2) {$X_3$};
  \node (n4) at (225:2) {$X_4$};

  \foreach \from/\to in {n1/n2,n2/n3,n3/n4}
    \draw (\from) -- (\to);
\end{tikzpicture}
\end{figure}
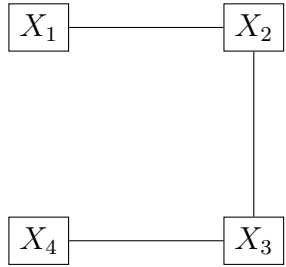

\begin{example} \label{Examp42}
Let $(N,R,M) \in \mathcal{R}^4$ be the ranking problem in Figure~\ref{Fig42}: $X_1$ has drawn against $X_2$, $X_2$ against $X_3$ and $X_3$ against $X_4$.
\end{example}

Theorem~\ref{Theo41} would be more straightforward as a strengthening of Proposition~\ref{Prop42} if self-consistency implies neutrality and/or symmetry. However, it is not the case as the following result holds.

\begin{remark} \label{Rem42}
There exists a scoring method that is self-consistent, but not neutral and symmetric.
\end{remark}

\begin{proof}
The statement can be verified by an example where an $SC$-compatible scoring method violates $NEU$ and $SYM$.

Consider Example~\ref{Examp42} with a scoring method $f$ such that $f_1(N,R,M) > f_2(N,R,M) > f_3(N,R,M) > f_4(N,R,M)$, for example, player $X_i$ gets the score $4-i$. $f$ meets self-consistency since $X_1$ has the same result against a stronger opponent compared to $X_4$, while there exists no correspondence between opponent sets $O_2$ and $O_3$ satisfying the conditions of $SC$.

Let $\sigma: N \to N$ be a permutation such that $\sigma(X_1) = X_4$, $\sigma(X_2) = X_3$, $\sigma(X_3) = X_2$, and $\sigma(X_4) = X_1$. Since $\sigma(N,R,M) = (N,R,M)$, $NEU$ implies $f_4(N,R,M) > f_1(N,R,M)$ and $f_3(N,R,M) > f_2(N,R,M)$, a contradiction.
Furthermore, $SYM$ leads to $f_1(N,R,M) = f_2(N,R,M) = f_3(N,R,M) = f_4(N,R,M)$, another impossibility.
Therefore there exists a self-consistent scoring method, which is not neutral and symmetric.
\end{proof}

\section{Conclusions} \label{Sec5}

We have found some unexpected implications of different properties in the case of generalized tournaments where the players should be ranked on the basis of their match results against each other.
First, self-consistency prohibits the use of individual scoring methods, that is, scores cannot be derived before the aggregation of tournament rounds (Proposition~\ref{Prop31}).
Second, independence of irrelevant matches (posing a kind of independence concerning the pairwise ranking of two players) follows from three axioms, neutrality (independence of relabelling the players), symmetry (implying a flat ranking if all aggregated comparisons are draws), and strong order preservation (perhaps the most natural property concerning the aggregation of ranking problems).
According to \citet{Csato2018f}, there exists no scoring method satisfying self-consistency and independence of irrelevant matches, hence Proposition~\ref{Prop41} implies that neutrality, symmetry, strong order preservation and self-consistency cannot be met simultaneously (Proposition~\ref{Prop42}). It even turns out that self-consistency and a weaker version of strong order preservation are still enough to derive this negative result (Theorem~\ref{Theo41}), consequently, one should choose between these two natural fairness requirements.

What do our results say to practitioners who want to rank players or teams?
First, self-consistency does not allow to rank them in individual rounds, one has to wait until all tournament results are known and can be aggregated.
Second, self-consistency is not compatible with order preservation on this universal domain. It is not an unexpected and counter-intuitive result as, according to \citet{Gonzalez-DiazHendrickxLohmann2013}, a number of ranking methods violate order preservation. We have proved that there is no hope to find a reasonable scoring method with this property.
From a more abstract point of view, breaking of order preservation in tournament ranking is a version of \href{https://en.wikipedia.org/wiki/Simpson\%27s_paradox}{Simpson's paradox}, a phenomenon in probability and statistics, in which a trend appears in different groups of data but disappears or reverses when these groups are combined.\footnote{~We are grateful to an anonymous referee for this remark.}
This negative result holds despite self-consistency is somewhat weaker than our intuition suggests: it does not imply neutrality and symmetry, so even a self-consistent ranking of players may depend on their names and without ties if all matches are drawn (Remark~\ref{Rem42}).
Third, losing the simplicity provided by order preservation certainly does not facilitate the axiomatic construction of scoring methods.

Consequently, while sacrificing self-consistency or order preservation seems to be unavoidable in our general setting, an obvious continuation of the current research is to get positive possibility results by some domain restrictions or further weakening of the axioms.
It is also worth to note that the incompatibility of the two axioms does not imply that any scoring method is always going to work badly, but all can lead to problematic results at times.


\section*{Acknowledgements}
\addcontentsline{toc}{section}{Acknowledgements}
\noindent
We are grateful to \emph{S\'andor Boz\'oki} for useful advice. \\
Anonymous reviewers provided valuable comments and suggestions on earlier drafts. \\
The research was supported by OTKA grant K 111797 and by the MTA Premium Post Doctorate Research Program. 


\end{document}